\documentclass{article}

\usepackage{caption}
\usepackage{subcaption}
\usepackage[final]{nips_2016}
\usepackage{amsmath,amssymb,rotating,multirow, amsthm}
\usepackage[utf8]{inputenc} 
\usepackage[T1]{fontenc}    
\usepackage{hyperref}       
\usepackage{url}            
\usepackage{booktabs}       
\usepackage{amsfonts}       
\usepackage{nicefrac}       
\usepackage{microtype,arydshln}      
\usepackage{graphicx} 
\usepackage{multirow}
\usepackage{mathrsfs}
\usepackage{amssymb,color,amsmath}
\usepackage{lscape, pbox, float, hyperref}
\usepackage{graphicx, listings, color}
\usepackage[utf8]{inputenc}
\usepackage[english]{babel}
\usepackage[normalem]{ulem}  
\newtheorem{lemma}{Lemma}

\usepackage[table]{xcolor}

\newcommand{\R}{\mathbb{R}}

\newtheorem{theorem}{Theorem}

\newcommand*\samethanks[1][\value{footnote}]{\footnotemark[#1]}

\title{Flexible Models for Microclustering with\\ Application to Entity Resolution}

\author{
Giacomo Zanella\thanks{Giacomo Zanella and Brenda Betancourt are joint first authors.}\\
Department of Decision Sciences\\
Bocconi University\\
\texttt{\small giacomo.zanella@unibocconi.it}\\
\And
Brenda Betancourt\samethanks\\
Department of Statistical Science\\
Duke University\\ 
\texttt{\small bb222@stat.duke.edu}\\
\AND
Hanna Wallach\\
Microsoft Research\\
\texttt{\small hanna@dirichlet.net}\\
\And
Jeffrey Miller\\
Department of Biostatistics\\
Harvard University\\
\texttt{\small jwmiller@hsph.harvard.edu}\\
\And
Abbas Zaidi\\
Department of Statistical Science\\
Duke University\\ 
\texttt{\small amz19@stat.duke.edu}\\
\And
Rebecca C. Steorts\\
Departments of Statistical Science and Computer Science\\
Duke University\\ 
\texttt{\small beka@stat.duke.edu}\\
}

\newcommand{\bpi}{\boldsymbol{\pi}}
\newcommand{\g}{\,|\,}
\newcommand{\clusters}{\boldsymbol{\kappa}}
\newcommand{\cluster}[1]{\kappa_{#1}}
\newcommand{\sizes}{\boldsymbol{\mu}}
\newcommand{\size}[1]{\mu_{#1}}

\newcommand{\edist}{\boldsymbol{\gamma}}
\newcommand{\shape}{\eta}
\newcommand{\rate}{s}
\newcommand{\betaA}{u}
\newcommand{\betaB}{v}

\newcommand{\teq}{\!=\!}
\newcommand\iid{\mathrel{\stackrel{\makebox[0pt]{\mbox{\normalfont\tiny
          \textrm{iid}}}}{\sim}}}
\renewcommand{\liminf}{\operatornamewithlimits{\textrm{lim\,inf}}}

\begin{document}

\maketitle

\begin{abstract}
Most generative models for clustering implicitly assume that the
number of data points in each cluster grows linearly with the total
number of data points. Finite mixture models, Dirichlet process
mixture models, and Pitman--Yor process mixture models make this
assumption, as do all other infinitely exchangeable clustering
models. However, for some applications, this assumption is
inappropriate. For example, when performing entity resolution, the
size of each cluster should be unrelated to the size of the data set,
and each cluster should contain a negligible fraction of the total
number of data points. These applications require models that yield
clusters whose sizes grow sublinearly with the size of the data
set. We address this requirement by defining the microclustering
property and introducing a new class of models that can exhibit this
property. We compare models within this class to two commonly used
clustering models using four entity-resolution data sets.
\end{abstract}

\section{Introduction}
\label{sec:introduction}

Many clustering applications require models that assume cluster sizes
grow linearly with the size of the data set. These applications
include topic modeling, inferring population structure, and
discriminating among cancer subtypes. Infinitely exchangeable
clustering models, including finite mixture models, Dirichlet process
mixture models, and Pitman--Yor process mixture models, all make this
linear-growth assumption, and have seen numerous successes when used
in these contexts. For other clustering applications, such as entity
resolution, this assumption is inappropriate. Entity resolution
(including record linkage and de-duplication) involves identifying
duplicate\footnote{In the entity resolution literature, the term
  ``duplicate records'' does not mean that the records are identical,
  but rather that the records are corrupted, degraded, or otherwise
  noisy representations of the same entity.}  records in noisy
databases~\cite{christen12data,christen12survey}, traditionally by
directly linking records to one another. Unfortunately, this
traditional approach is computationally infeasible for large data
sets---a serious limitation in ``the age of big
data''~\cite{christen12data,winkler06overview}. As a result,
researchers increasingly treat entity resolution as a clustering
problem, where each entity is implicitly associated with one or more
records and the inference goal is to recover the latent entities
(clusters) that correspond to the observed records (data
points)~\cite{steorts??bayesian,steorts15entity,steorts14smered}. In
contrast to other clustering applications, the number of data points
in each cluster should remain small, even for large data
sets. Applications like this require models that yield clusters whose
sizes grow sublinearly with the total number of data points
\cite{broderick14variational}. To address this requirement, we
define the microclustering property in
section~\ref{sec:microclustering} and, in section~\ref{sec:micro},
introduce a new class of models that can exhibit this property. In
section~\ref{sec:experiments}, we compare two models within this class
to two commonly used infinitely exchangeable clustering
models.\looseness=-1

\section{The Microclustering Property}
\label{sec:microclustering}

To cluster $N$ data points $x_1, \ldots, x_N$ using a partition-based
Bayesian clustering model, one first places a prior over partitions of
$[N] = \{ 1, \ldots, N \}$. Then, given a partition $C_N$ of $[N]$,
one models the data points in each part $c \in C_N$ as jointly
distributed according to some chosen distribution. Finally, one
computes the posterior distribution over partitions and, e.g., uses it
to identify probable partitions of $[N]$. Mixture models are a
well-known type of partition-based Bayesian clustering model, in which
$C_N$ is implicitly represented by a set of cluster assignments $z_1,
\ldots, z_N$. These cluster assignments can be regarded as the first
$N$ elements of an infinite sequence $z_1, z_2, \ldots$, drawn a
priori from\looseness=-1
\begin{equation}
  \bpi \sim H \quad \textrm{and} \quad
 z_1, z_2, \ldots \g \bpi \iid \bpi,
 \label{eqn:mixture}
\end{equation}
where $H$ is a prior over $\bpi$ and $\bpi$ is a vector of mixture
weights with $\sum_l \pi_l \teq 1$ and $\pi_l \geq 0$ for
all~$l$. Commonly used mixture models include (a) finite mixtures
where the dimensionality of $\bpi$ is fixed and $H$ is usually a
Dirichlet distribution; (b) finite mixtures where the dimensionality
of $\bpi$ is a random
variable~\cite{richardson97bayesian,miller15mixture}; (c) Dirichlet
process (DP) mixtures where the dimensionality of $\bpi$ is
infinite~\cite{sethuraman94constructive}; and (d) Pitman--Yor process
(PYP) mixtures, which generalize DP
mixtures~\cite{ishwaran03generalized}.

Equation~\ref{eqn:mixture} implicitly defines a prior over partitions
of $\mathbb{N} = \{ 1, 2, \ldots \}$. Any random partition
$C_{\mathbb{N}}$ of $\mathbb{N}$ induces a sequence of random
partitions $(C_N : N=1, 2, \ldots)$, where $C_N$ is a partition of
$[N]$. Via the strong law of large numbers, the cluster sizes in any
such sequence obtained via equation~\ref{eqn:mixture} grow linearly
with $N$ because, with probability one, for all $l$, $\frac{1}{N}
\sum_{n=1}^N I(z_n \teq l) \rightarrow \pi_l$ as $N \rightarrow
\infty$, where $I(\cdot)$ denotes the indicator
function. Unfortunately, this linear growth assumption is not
appropriate for entity resolution and other applications that require
clusters whose sizes grow sublinearly with $N$.

To address this requirement, we therefore define the microclustering
property: A sequence of random partitions $(C_N : N=1, 2, \ldots)$
exhibits the microclustering property if $M_N$ is $o_p(N)$, where
$M_N$ is the size of the largest cluster in $C_N$, or, equivalently,
if $M_N \,/\, N \rightarrow 0$ in probability as $N \rightarrow
\infty$.

A clustering model exhibits the microclustering property if the
sequence of random partitions implied by that model satisfies the
above definition. No mixture model can exhibit the microclustering
property (unless its parameters are allowed to vary with $N$). In
fact, Kingman's paintbox
theorem~\cite{kingman78representation,aldous85exchangeability} implies
that any exchangeable partition of $\mathbb{N}$, such as a partition
obtained using equation~\ref{eqn:mixture}, is either equal to the
trivial partition in which each part contains one element or satisfies
$\liminf_{N \rightarrow \infty} M_N \,/\, N > 0$ with positive
probability. By Kolmogorov's extension theorem, a sequence of random
partitions $(C_N : N=1, 2, \ldots)$ corresponds to an exchangeable
random partition of $\mathbb{N}$ whenever (a) each $C_N$ is finitely
exchangeable (i.e., its probability is invariant under permutations of
$\{1,\dots,N\}$) and (b) the sequence is projective (also known as
consistent in distribution)---i.e., if $N' \!<\! N$, the distribution
over $C_{N'}$ coincides with the marginal
distribution over partitions of $[N']$ induced by the distribution over $C_N$.
Therefore, to obtain a nontrivial model that
exhibits the microclustering property, we must sacrifice either (a)
or (b). Previous work~\cite{wallach10alternative} sacrificed (a);
in this paper, we instead sacrifice (b). 

Sacrificing finite exchangeability and sacrificing projectivity have
very different consequences. If a partition-based Bayesian clustering
model is not finitely exchangeable, then inference will depend on the
order of the data points. For most applications, this consequence is
undesirable---there is no reason to believe that the order of the data
points is meaningful. In contrast, if a model lacks projectivity, then
the implied joint distribution over a subset of the data points in a
data set will not be the same as the joint distribution obtained by
modeling the subset directly. In the context of entity resolution,
sacrificing projectivity is a more natural and less restrictive choice
than sacrificing finite exchangeability.\looseness=-1

\section{Kolchin Partition Models for Microclustering}
\label{sec:micro}

We introduce a new class of Bayesian models for microclustering by
placing a prior on the number of clusters $K$ and, given $K$, modeling
the cluster sizes $N_1, \ldots, N_K$ directly.  We start by defining
\begin{equation}
  K \sim \clusters \quad \textrm{and} \quad
 N_1,\ldots, N_K \g K \iid \sizes,
\label{eq:FMMC}
\end{equation}
where $\clusters=(\cluster{1},\cluster{2},\dots)$ and
$\sizes=(\size{1},\size{2},\dots)$ are probability distributions over
$\mathbb{N} = \{1,2,\ldots\}$. We then define $N = \sum_{k=1}^K N_k$
and, given $N_1, \ldots, N_K$, generate a set of cluster assignments
$z_1, \ldots, z_N$ by drawing a vector uniformly at random from the
set of permutations of $(\underbrace{1,\ldots,1}_\text{$N_1$
  times},\underbrace{2,\ldots,2}_\text{$N_2$
  times},\ldots\ldots,\underbrace{K,\ldots,K}_\text{$N_K$ times})$.
The cluster assignments $z_1, \ldots, z_N$ induce a random partition
$C_N$ of $[N]$, where $N$ is itself a random variable---i.e., $C_N$ is
a random partition of a random number of elements. We refer to the
resulting class of marginal distributions over $C_N$ as Kolchin
partition (KP) models~\citep{kolchin71problem,pitman06combinatorial}
because the form of equation~\ref{eq:FMMC} is closely related to
Kolchin's representation theorem for Gibbs-type partitions (see, e.g.,
\citealp[theorem 1.2]{pitman06combinatorial}). For appropriate choices
of $\clusters$ and $\sizes$, KP models can exhibit the microclustering
property (see appendix B for an example).

If $\mathscr{C}_N$ denotes the set of all possible partitions
of $[N]$, then $\bigcup_{N=1}^{\infty} \mathscr{C}_N$ is the set of
all possible partitions of $[N]$ for all $N\in\mathbb{N}$.
The probability of any given partition $C_N\in\bigcup_{N=1}^{\infty}
\mathscr{C}_N$ is
\begin{equation}
  \label{eq:size_models_pmf}
P(C_N) =
\frac{|C_N|!\,\cluster{|C_N|}}{N!}\left(\prod_{c \in C_N}
|c|!\,\size{|c|}\right),
\end{equation}
where $|\cdot|$ denotes the cardinality of a set, $|C_N|$ is the
number of clusters in $C_N$, and $|c|$ is the number of elements in
cluster c. In practice, however, $N$ is usually observed. Conditioned
on $N$, a KP model implies that $P(C_N \g N) \propto
|C_N|!\,\cluster{|C_N|}\left(\prod_{c \in C_N}
|c|!\,\size{|c|}\right)$. Equation \ref{eq:size_models_pmf} leads to a
``reseating algorithm''---much like the Chinese restaurant process
(CRP)---derived by sampling from $P(C_N \g N, C_N \!\setminus\! n)$,
where $C_N \!\setminus\! n$ is the partition obtained by removing
element $n$ from $C_N$:
\begin{itemize}
\item for $n=1,\ldots, N$, reassign element $n$ to
\begin{itemize}
\item an existing cluster $c \in C_{N}\!\setminus\! n$ with
  probability $\propto
  \left(|c|+1\right)\frac{\size{(|c|+1)}}{\size{|c|}}$
\item or a new cluster with probability $\propto
\left(|C_N \!\setminus\! n|+1\right)\tfrac{\cluster{(|C_N
    \!\setminus\! n|+1)}}{\cluster{|C_N \!\setminus\! n|}}\size{1}$.
\end{itemize}
\end{itemize}
We can use this reseating algorithm to draw samples from $P(C_N \g
N)$; however, unlike the CRP, it does not produce an exact sample if
it is used to incrementally construct a partition from the empty
set. In practice, this limitation does not lead to any negative
consequences because standard posterior inference sampling methods do
not rely on this property. When a KP model is used as the prior in a
partition-based clustering model---e.g., as an alternative to
equation~\ref{eqn:mixture}---the resulting Gibbs sampling algorithm
for $C_N$ is similar to this reseating algorithm, but accompanied by
likelihood terms. Unfortunately, this algorithm is slow for large data
sets. In appendix C, we therefore propose a faster Gibbs sampling
algorithm---the chaperones algorithm---that is particularly well
suited to microclustering.\looseness=-1

In sections~\ref{sec:nbnb} and~\ref{sec:nbd}, we introduce two related
KP models for microclustering, and in section~\ref{sec:application} we
explain how KP models can be applied in the context of entity
resolution with categorical data.

\subsection{The NBNB Model}
\label{sec:nbnb}

We start with equation~\ref{eq:size_models_pmf} and define 
\begin{equation}
\label{eqn:nbnb}
\clusters = \textrm{NegBin}\left(a,q\right) \quad \textrm{and} \quad \sizes =
\textrm{NegBin}\left(r,p\right),
\end{equation}
where $\textrm{NegBin}(a,q)$ and $\textrm{NegBin}(r,p)$ are negative
binomial distributions truncated to $\mathbb{N} = \{1,2,\dots\}$. We
assume that $a \!>\! 0$ and $q \!\in\! (0, 1)$ are fixed
hyperparameters, while $r$ and $p$ are distributed as $r \sim
\textrm{Gam}(\shape_r, \rate_r)$ and $p \sim
\textrm{Beta}(\betaA_p,\betaB_p)$ for fixed $\shape_r$, $\rate_r$,
$\betaA_p$ and $\betaB_p$.\footnote{We use the shape-and-rate
  parameterization of the gamma distribution.} We refer to the
resulting marginal distribution over $C_N$ as the negative
binomial--negative binomial (NBNB) model.

By substituting equation~\ref{eqn:nbnb} into equation
\ref{eq:size_models_pmf}, we obtain the probability of $C_N$
conditioned $N$:
\begin{equation}
  \label{eqn:NBNB_conditional}
P(C_N \g N,a, q, r,p) \propto \Gamma\left(|C_N|+a\right) \beta^{|C_N|}
\prod_{c \in C_N}
\frac{\Gamma\left(|c|+r\right)}{\Gamma\left(r\right)}\,,
\end{equation}
where $\beta=\frac{q\,(1-p)^{r}}{1-(1-p)^r}$. We provide the complete
derivation of equation~\ref{eqn:NBNB_conditional}, along with the
conditional posterior distributions over $r$ and $p$, in appendix
A.2. Posterior inference for the NBNB model involves alternating
between (a) sampling $C_N$ from $P(C_N \g N, a, q, r, p)$ using the
chaperones algorithm and (b) sampling $r$ and $p$ from their
respective conditional posteriors using, e.g., slice
sampling~\cite{neal03slice}.

\subsection{The NBD Model}
\label{sec:nbd}

Although $\clusters = \textrm{NegBin}\left(a,q\right)$ will yield
plausible values of $K$, $\sizes = \textrm{NegBin}\left(r, p\right)$
may not be sufficiently flexible to capture realistic properties of
$N_1, \ldots, N_K$, especially when $K$ is large. For example, in a
record-linkage application involving two otherwise noise-free
databases containing thousands of records, $K$ will be large and each
$N_k$ will be at most two. A negative binomial distribution cannot
capture this property. We therefore define a second KP model---the
negative binomial--Dirichlet (NBD) model---by taking a nonparametric
approach to modeling $N_1, \ldots, N_K$ and drawing $\sizes$ from an
infinite-dimensional Dirichlet distribution over the positive
integers:
\begin{equation}
\label{eqn:nbd}
\clusters = \textrm{NegBin}\left(a,q\right) \quad \text{and} \quad
\sizes  \g \alpha, \sizes^{(0)} \sim
\textrm{Dir}\left(\alpha, \sizes^{(0)}\right),
\end{equation}
where $\alpha>0$ is a fixed concentration parameter and
$\sizes^{(0)}=(\mu^{(0)}_1,\mu^{(0)}_2,\cdots)$ is a fixed base
measure with $\sum_{m=1}^\infty \mu^{(0)}_m=1$ and $\mu^{(0)}_m\geq 0$
for all $m$. The probability of $C_N$ conditioned on $N$ and $\sizes$
is
\begin{equation}
  \label{eq:C_N_for_NBD}
P(C_N \g N,a, q, \sizes) \propto \Gamma\left(|C_N|+a\right) q^{|C_N|}
\prod_{c \in C_N} |c|!\,\size{|c|}.
\end{equation}
Posterior inference for the NBD model involves alternating between (a)
sampling $C_N$ from $P(C_N \g N, a, q, \sizes)$ using the chaperones
algorithm and (b) sampling $\sizes$ from its conditional
posterior:\looseness=-1
\begin{equation}
  \label{eqn:sizes_cond_posterior}
\sizes \g C_N, \alpha, \boldsymbol{\mu}^{(0)} \sim
\textrm{Dir}\left(\alpha\,\size{1}^{(0)}+L_{1},\alpha\,\size{2}^{(0)}+L_{2},\dots\right),
\end{equation}
where $L_m$ is the number of clusters of size $m$ in $C_N$. Although
$\sizes$ is an infinite-dimensional vector, only the first $N$
elements affect $P(C_N \g a, q, \sizes)$. Therefore, it is sufficient
to sample the $(N+1)$-dimensional vector $(\size{1}, \ldots, \size{N},
1 - \sum_{m=1}^N \size{m})$ from
equation~\ref{eqn:sizes_cond_posterior}, modified accordingly, and
retain only $\size{1}, \ldots, \size{N}$. We provide complete
derivations of equations~\ref{eq:C_N_for_NBD}
and~\ref{eqn:sizes_cond_posterior} in appendix A.3.

\subsection{The Microclustering Property for the NBNB and NBD Models}
\label{sec:appendix_c}

Figure~\ref{fig:microclustering} contains empirical evidence
suggesting that the NBNB and NBD models both exhibit the
microclustering property. For each model, we generated samples of $M_N
\,/\, N$ for $N = 100, \ldots, 10^4$. For the NBNB model, we set
$a=1$, $q=0.5$, $r=1$, and $p=0.5$ and generated the samples using
rejection sampling. For the NBD model, we set $a=1$, $q=0.5$, and
$\alpha=1$ and set $\sizes^{(0)}$ to be a geometric distribution over
$\mathbb{N} = \{1, 2, \ldots\}$ with a parameter of 0.5. We generated
the samples using MCMC methods. For both models, $M_N \,/\, N$ appears
to converge to zero in probability as $N \rightarrow \infty$, as
desired.\looseness=-1

In appendix B, we also prove that a variant of the NBNB model exhibits
the microclustering property.

\begin{figure}[t]
\centering
\vspace{-4mm}
\includegraphics[width=0.75\textwidth]{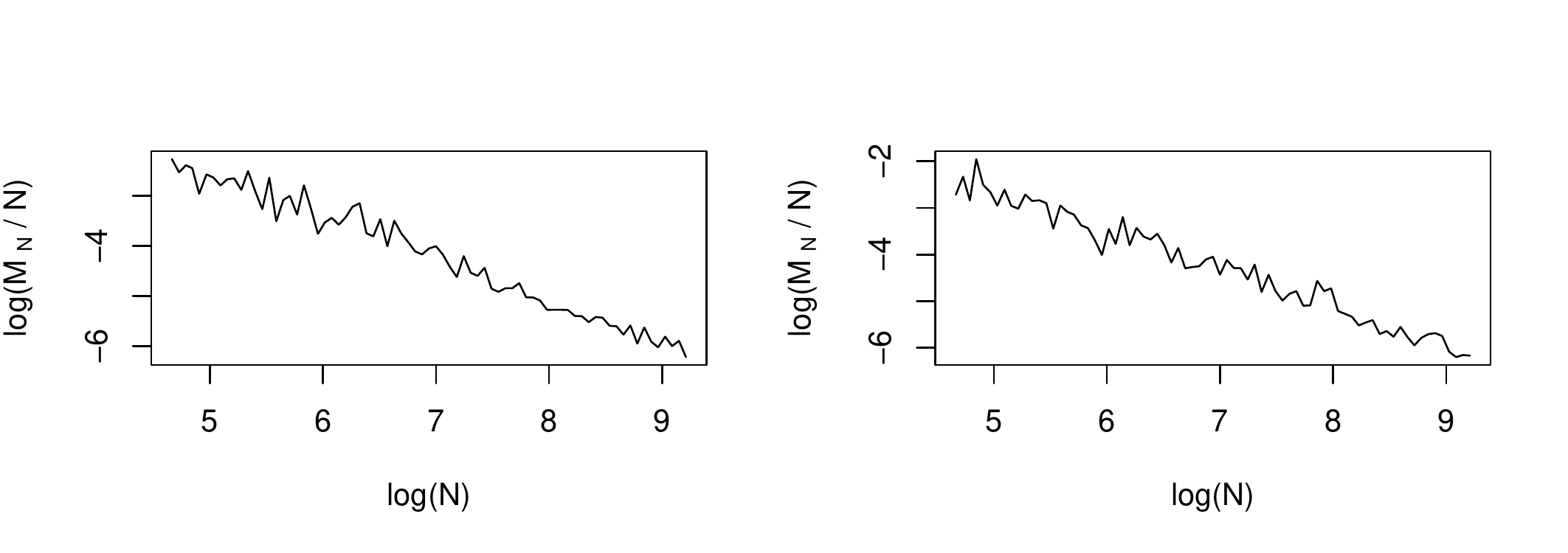}
\vspace{-2mm}
\caption{The NBNB (left) and NBD (right) models appear to exhibit
  the microclustering property.\looseness=-1}
\vspace{-4mm}
\label{fig:microclustering}
\end{figure}

\subsection{Application to Entity Resolution}
\label{sec:application}

KP models can be used to perform entity resolution. In this context,
the data points $x_1, \ldots, x_N$ are observed records and the $K$
clusters are latent entities. If each record consists of $F$
categorical fields, then\looseness=-1
\begin{align}
  C_N &\sim \textrm{KP model}\\
  \label{eqn:theta}
  \boldsymbol{\theta}_{fk} \g \delta_f, \edist_f &\sim
  \textrm{Dir}\left(\delta_f, \edist_f \right)\\
  z_n &\sim \zeta(C_N, n)\\
  \label{eqn:x}
  x_{fn} \g z_n, \boldsymbol{\theta}_{f1}, \ldots,
  \boldsymbol{\theta}_{fK} &\sim
  \textrm{Cat}\left(\boldsymbol{\theta}_{f z_n}\right)
\end{align}
for $f=1, \ldots, F$, $k = 1, \ldots, K$, and $n = 1, \ldots, N$,
where $\zeta(C_N, n)$ maps the $n^{\textrm{th}}$ record to a latent
cluster assignment $z_n$ according to $C_N$. We assume that $\delta_f
> 0$ is distributed as $\delta_f \sim \textrm{Gam}\left(1, 1\right)$,
while $\edist_f$ is fixed. Via Dirichlet--multinomial conjugacy, we
can marginalize over $\boldsymbol{\theta}_{11}, \ldots,
\boldsymbol{\theta}_{FK}$ to obtain a closed-form expression for
$P(x_1, \ldots, x_N \g z_1, \ldots, z_N, \delta_f,
\edist_f)$. Posterior inference involves alternating between (a)
sampling $C_N$ from $P(C_N \g x_1, \ldots, x_N, \delta_f)$ using the
chaperones algorithm accompanied by appropriate likelihood terms, (b)
sampling the parameters of the KP model from their conditional
posteriors, and (c) sampling $\delta_f$ from its conditional posterior
using slice sampling.\looseness=-1

\section{Experiments} 
\label{sec:experiments}

In this section, we compare two entity resolution models based on the
NBNB model and the NBD model to two similar models based on the DP
mixture model~\cite{sethuraman94constructive} and the PYP mixture
model~\cite{ishwaran03generalized}. All four models use the likelihood
in equations~\ref{eqn:theta} and~\ref{eqn:x}. For the NBNB model and
the NBD model, we set $a$ and $q$ to reflect a weakly informative
prior belief that
$\mathbb{E}[K]=\sqrt{\textrm{Var}[K]}=\frac{N}{2}$. For the NBNB
model, we set $\shape_r = \rate_r = 1$ and
$\betaA_p=\betaB_p=2$.\footnote{We used $p \sim \textrm{Beta}\left(2,
  2\right)$ because a uniform prior implies an unrealistic prior
  belief that $\mathbb{E}[N_k] = \infty$.\looseness=-1} For the NBD
model, we set $\alpha = 1$ and set $\sizes^{(0)}$ to be a geometric
distribution over $\mathbb{N} = \{1, 2, \ldots\}$ with a parameter of
0.5. This base measure reflects a prior belief that $\mathbb{E}[N_k] =
2$. Finally, to ensure a fair comparison between the two different
classes of model, we set the DP and PYP concentration parameters to
reflect a prior belief that $\mathbb{E}[K] = \frac{N}{2}$.

We assess how well each model ``fits'' four data sets typical of those
arising in real-world entity resolution applications. For each data
set, we consider four statistics: (a) the number of singleton
clusters, (b) the maximum cluster size, (c) the mean cluster size, and
(d) the 90$^{\textrm{th}}$ percentile of cluster sizes. We compare
each statistic's true value to its posterior distribution according to
each of the models. For each model and data set combination, we also
consider five entity-resolution summary statistics: (a) the posterior
expected number of clusters, (b) the posterior standard error, (c) the
false negative rate, (d) the false discovery rate, and (e) the
posterior expected value of $\delta_f = \delta$ for $f=1, \ldots,
F$. The false negative and false discovery rates are both invariant
under permutations of $1, \ldots, K$~\cite{steorts15entity,
  steorts14comparison}.\looseness=-1

\subsection{Data Sets}
\label{sec:data}

We constructed four realistic data sets, each consisting of $N$
records associated with $K$ entities.


\textbf{Italy:} We derived this data set from the Survey on Household
Income and Wealth, conducted by the Bank of Italy every two
years. There are nine categorical fields, including year of birth,
employment status, and highest level of education attained. Ground
truth is available via unique identifiers based upon social security
numbers; roughly 74\% of the clusters are singletons. We used the 2008
and 2010 databases from the Fruili region to create a record-linkage
data set consisting of $N=789$ records; each $N_k$ is at most two. We
discarded the records themselves, but preserved the number of fields,
the empirical distribution of categories for each field, the number of
clusters, and the cluster sizes. We then generated synthetic records
using equations~\ref{eqn:theta} and~\ref{eqn:x}. We created three
variants of this data set, corresponding to $\delta = 0.02, 0.05,
0.1$. For all three, we used the empirical distribution of categories
for field $f$ as $\edist_f$. By generating synthetic records in this
fashion, we preserve the pertinent characteristics of the original
data, while making it easy to isolate the impacts of the different
priors over partitions.\looseness=-1

\textbf{NLTCS5000:} We derived this data set from the National Long
Term Care Survey
(NLTCS)\footnote{\url{http://www.nltcs.aas.duke.edu/}}---a
longitudinal survey of older Americans, conducted roughly every six
years. We used four of the available fields: date of birth, sex, state
of residence, and regional office. We split date of birth into three
separate fields: day, month, and year. Ground truth is available via
social security numbers; roughly 68\% of the clusters are
singletons. We used the 1982, 1989, and 1994 databases and
down-sampled the records, preserving the proportion of clusters of
each size and the maximum cluster size, to create a record-linkage
data set of $N=5,000$ records; each $N_k$ is at most three. We then
generated synthetic records using the same approach that we used to
create the Italy data set.\looseness=-1

\textbf{Syria2000 and SyriaSizes:} We constructed these data sets from
data collected by four human-rights groups between 2011 and 2014 on
people killed in the Syrian
conflict~\cite{price13updated,price14updated}. Hand-matched ground
truth is available from the Human Rights Data Analysis Group. Because
the records were hand matched, the data are noisy and potentially
biased. Performing entity resolution is non-trivial because there are
only three categorical fields: gender, governorate, and date of
death. We split date of death, which is present for most records, into
three separate fields: day, month, and year. However, because the
records only span four years, the year field conveys little
information. In addition, most records are male, and there are only
fourteen governorates. We created the Syria2000 data set by
down-sampling the records, preserving the proportion of clusters of
each size, to create a data set of $N=2,000$ records; the maximum
cluster size is five. We created the SyriaSizes data set by
down-sampling the records, preserving some of the larger clusters
(which necessarily contain within-database duplications), to create a
data set of $N=6,700$ records; the maximum cluster size is ten. We
provide the empirical distribution over cluster sizes for each data
set in appendix D. We generated synthetic records for both data sets
using the same approach that we used to create the Italy data
set.\looseness=-1

\subsection{Results}
\label{sec:results}

We report the results of our experiments in table~\ref{tab:1} and
figure~\ref{fig:1}. The NBNB and NBD models outperformed the DP and
PYP models for almost all variants of the Italy and NLTCS5000 data
sets. In general, the NBD model performed the best of the four, and
the differences between the models' performance grew as the value of
$\delta$ increased. For the Syria2000 and SyriaSizes data sets, we see
no consistent pattern to the models' abilities to recover the true
values of the data-set statistics. Moreover, all four models had poor
false negative rates, and false discovery rates---most likely because
these data sets are extremely noisy and contain very few fields. We
suspect that no entity resolution model would perform well for these
data sets. For three of the four data sets, the exception being the
Syria2000 data set, the DP model and the PYP model both greatly
overestimated the number of clusters for larger values of
$\delta$. Taken together, these results suggest that the flexibility
of the NBNB and NBD models make them more appropriate choices for most
entity resolution applications.

\begin{figure*}[htp]
  \centering
      {\makebox[\textwidth]{
          \begin{subfigure}{\textwidth}
            \centering
            \includegraphics[width=0.9\textwidth,height=150pt]{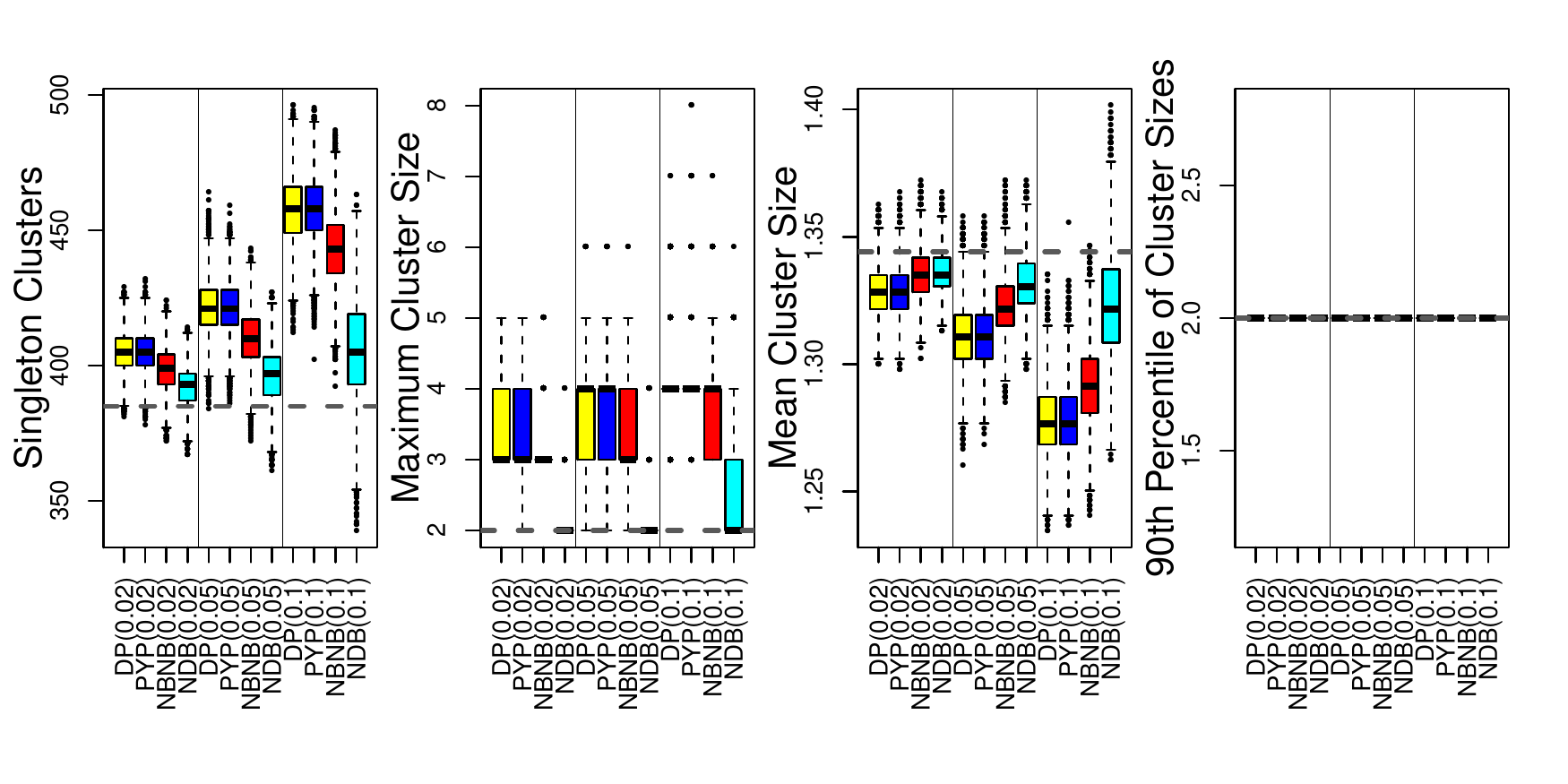}
            \vspace{-4mm}
            \caption{Italy: NBD model > NBNB model > PYP mixture model
              > DP mixture model.}
          \end{subfigure}  
        }
        \vspace{-2mm}
      }
      {\makebox[\textwidth]{
          \begin{subfigure}{\textwidth}
            \centering    
            \includegraphics[width=0.9\textwidth,height=150pt]{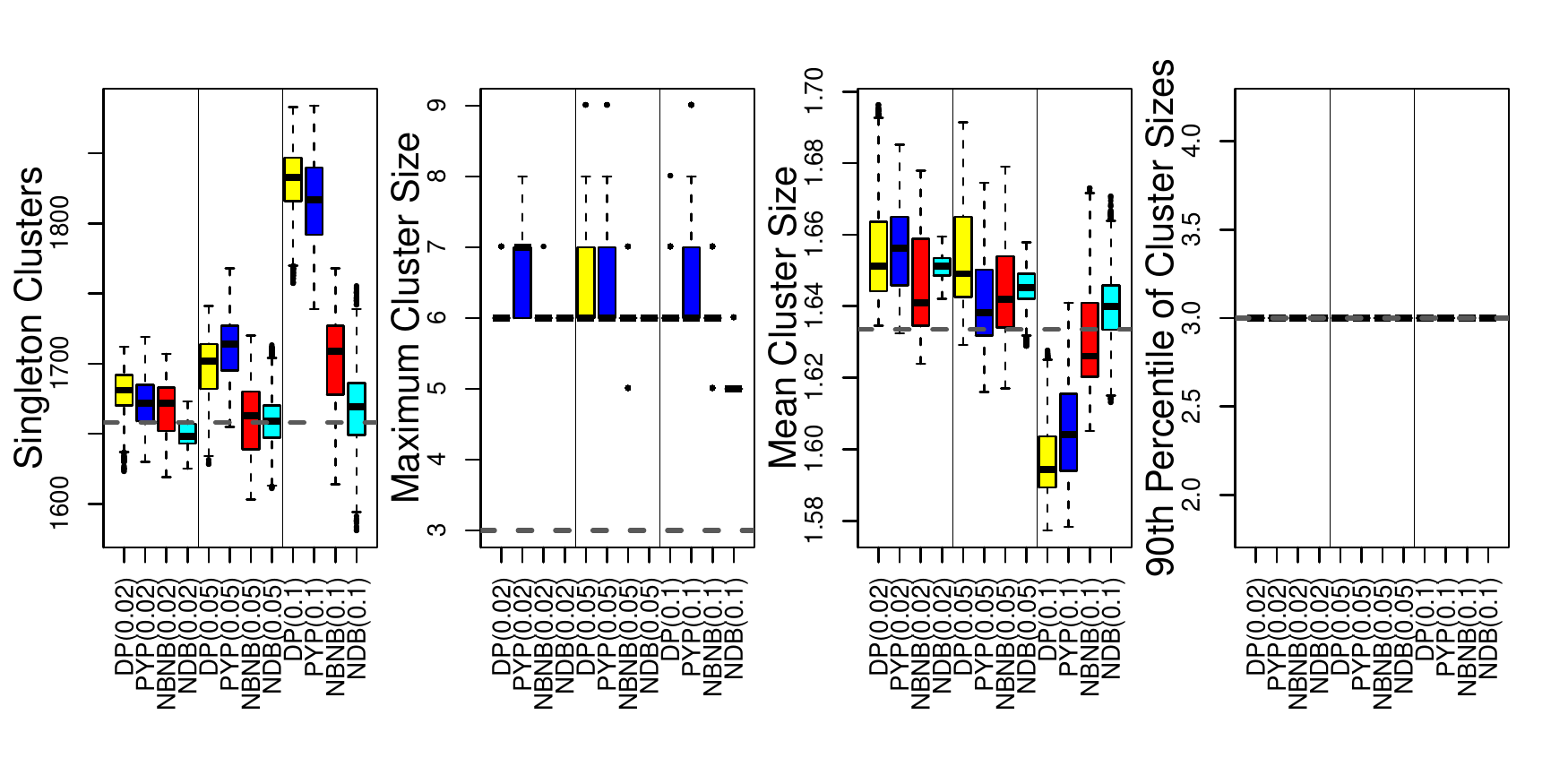}
            \vspace{-4mm}
            \caption{NLTCS5000: NBD model > NBNB model > PYP mixture
              model > DP mixture model.}
          \end{subfigure}
          }
        \vspace{-2mm}
      }
      {\makebox[\textwidth]{
          \begin{subfigure}{\textwidth}
            \centering    
            \includegraphics[width=0.9\textwidth,height=150pt]{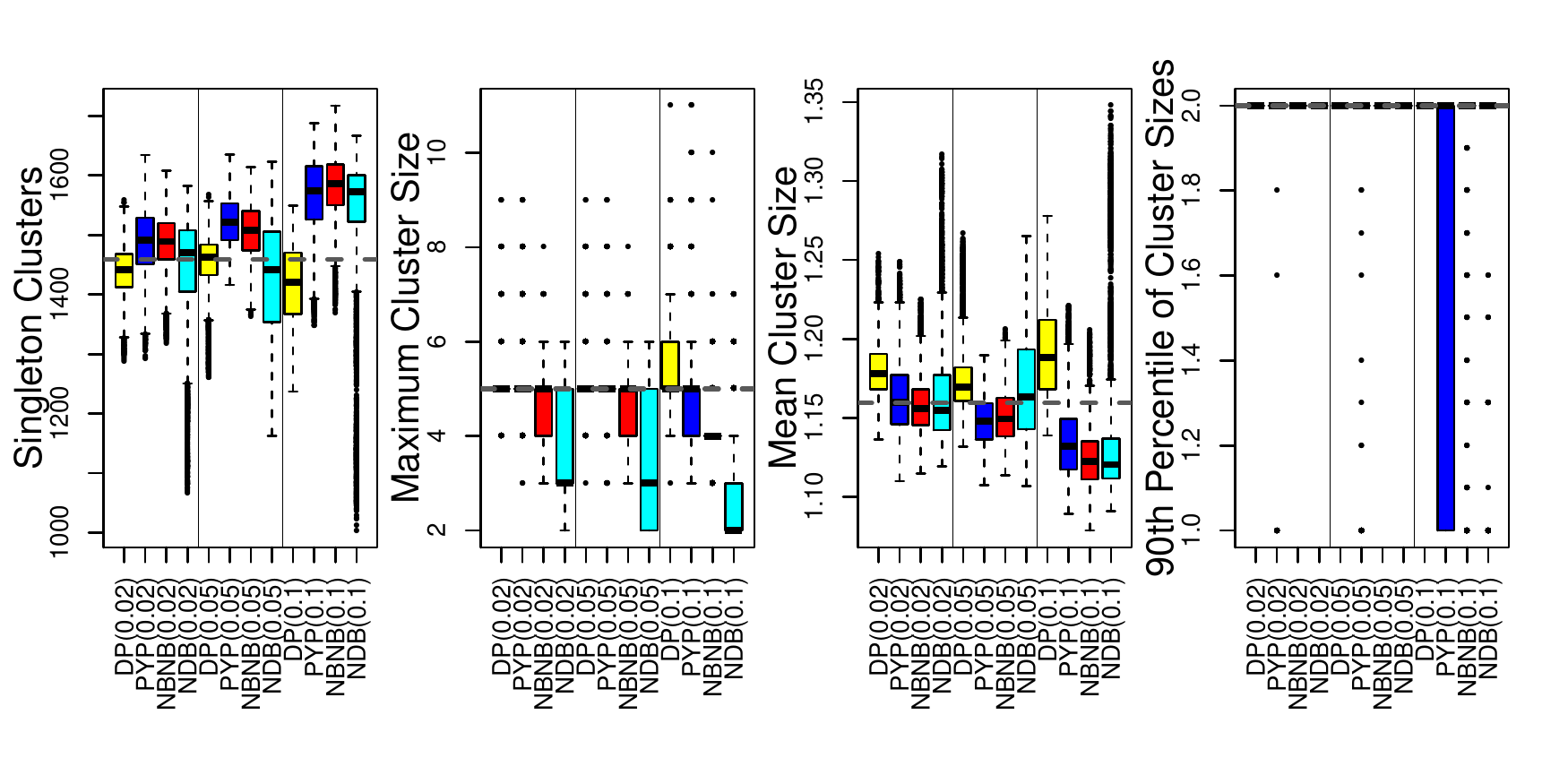}
            \vspace{-4mm}
            \caption{Syria2000: the models perform similarly because there
              are so few fields.}
          \end{subfigure}
        }
        \vspace{-2mm}
      }
{\makebox[\textwidth]{
    \begin{subfigure}{\textwidth}
      \centering    
      \includegraphics[width=0.9\textwidth,height=150pt]{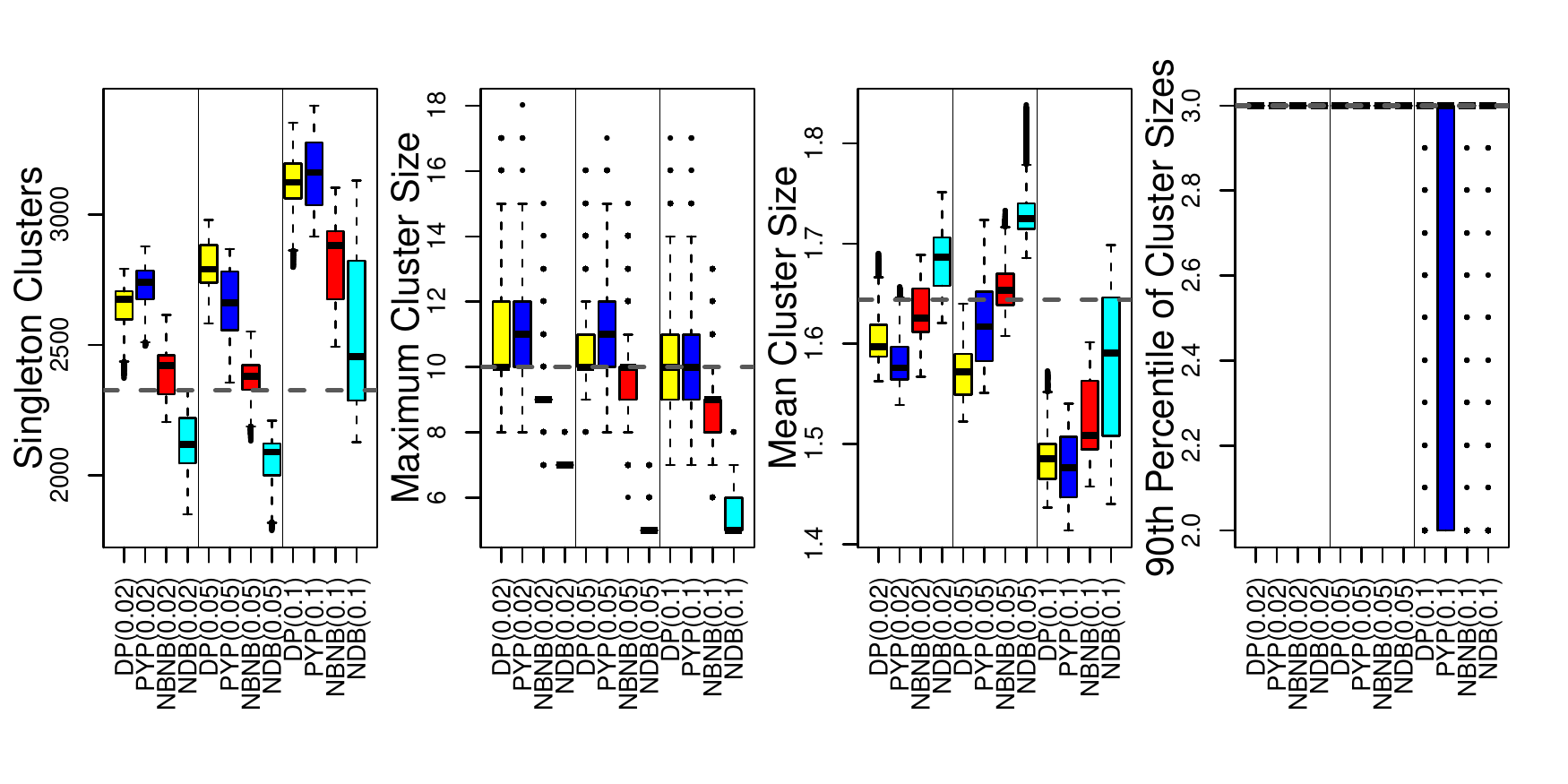}
      \vspace{-4mm}
      \caption{SyriaSizes: the models perform similarly because there
        are so few fields.}
    \end{subfigure}
  }
}
\caption{Box plots depicting the true value (dashed line) of each
  data-set statistic for each variant of each data set, as well as its
  posterior distribution according to each of the four entity
  resolution models.\looseness=-1}
\label{fig:1}
\end{figure*}

\begin{table}[htp]
\caption{Entity-resolution summary statistics---the posterior expected
  number of clusters, the posterior standard error, the false negative
  rate (lower is better), the false discovery rate (lower is better),
  and the posterior expected value of $\delta$---for each variant of
  each data set and each of the four models.\looseness=-1}
\vspace{2.5mm}
\footnotesize
\centering
\begin{tabular}{cccrccccc}
\toprule
\normalsize{Data Set} & \normalsize{True $K$} & \normalsize{Variant} & \normalsize{Model} & \normalsize{$\mathbb{E}[K]$} & \normalsize{Std. Err.} & \normalsize{FNR} & \normalsize{FDR} & \normalsize{$\mathbb{E}[\delta]$} \\ 
\midrule
Italy & 587 & $\delta = 0.02$ & DP & 594.00 & 4.51 & 0.07 & 0.03 & 0.02 \\
 & & & PYP & 593.90 & 4.52 & 0.07 & 0.03 & 0.02  \\
 & & & NBNB & 591.00 & 4.43 & 0.04 & 0.03 & 0.02 \\
 & & & NBD & 590.50 & 3.64 & 0.03 & 0.00 & 0.02 \\
\cmidrule{3-9}
 & & $\delta = 0.05$ & DP & 601.60 & 5.89 & 0.13 & 0.03 & 0.03 \\
 & & & PYP & 601.50 & 5.90 & 0.13 & 0.03 & 0.04 \\
 & & & NBNB & 596.40 & 5.79 & 0.11 & 0.04 & 0.04  \\
 & & & NBD & 592.60 & 5.20 & 0.09 & 0.04 & 0.04 \\
\cmidrule{3-9}
 & & $\delta = 0.1$ & DP & 617.40 & 7.23 & 0.27 & 0.06 & 0.07 \\
 & & & PYP & 617.40 & 7.22 & 0.27 & 0.05 & 0.07 \\
 & & & NBNB & 610.90 & 7.81 & 0.24 & 0.06 & 0.08 \\
 & & & NBD & 596.60 & 9.37 & 0.18 & 0.05 & 0.10 \\
\midrule
NLTCS5000 & 3,061 & $\delta = 0.02$ & DP & 3021.70 & 24.96 & 0.02 & 0.11 & 0.03  \\
& & & PYP & 3018.70 & 25.69 & 0.03 & 0.11 & 0.03 \\
& & & NBNB & 3037.80 & 25.18 & 0.02 & 0.07 & 0.02 \\
& & & NBD & 3028.20 & 5.65 & 0.01 & 0.09 & 0.03 \\
\cmidrule{3-9}
& & $\delta = 0.05$ & DP & 3024.00 & 26.15 & 0.05 & 0.13 & 0.06 \\
& & & PYP & 3045.80 & 23.66 & 0.05 & 0.10 & 0.05 \\
& & & NBNB & 3040.90 & 24.86 & 0.04 & 0.06 & 0.05 \\
& & & NBD & 3039.30 & 10.17 & 0.03 & 0.07 & 0.06 \\
\cmidrule{3-9}
& & $\delta = 0.1$ & DP & 3130.50 & 21.44 & 0.12 & 0.09 & 0.10 \\
& & & PYP & 3115.10 & 25.73 & 0.13 & 0.10 & 0.10\\
& & & NBNB & 3067.30 & 25.31 & 0.11 & 0.08 & 0.11 \\
& & & NBD & 3049.10 & 16.48 & 0.09 & 0.08 & 0.12\\
\midrule
Syria2000 & 1,725 & $\delta = 0.02$ & DP & 1695.20 & 25.40 & 0.70 & 0.27 & 0.07 \\
& & & PYP & 1719.70 & 36.10 & 0.71 & 0.26 & 0.04 \\
& & & NBNB & 1726.80 & 27.96 & 0.70 & 0.28 & 0.05 \\
& & & NBD & 1715.20 & 51.56 & 0.67 & 0.28 & 0.02 \\
\cmidrule{3-9}
& & $\delta = 0.05$ & DP & 1701.80 & 31.15 & 0.77 & 0.31 & 0.07  \\
& & & PYP & 1742.90 & 24.33 & 0.75 & 0.32 & 0.04  \\
& & & NBNB & 1738.30 & 25.48 & 0.74 & 0.31 & 0.04 \\
& & & NBD & 1711.40 & 47.10 & 0.69 & 0.32 & 0.03 \\
\cmidrule{3-9}
& & $\delta = 0.1$ & DP & 1678.10 & 40.56 & 0.81 & 0.19 & 0.18 \\
& & & PYP & 1761.20 & 39.38 & 0.81 & 0.22 & 0.08 \\
& & & NBNB & 1779.40 & 29.84 & 0.77 & 0.26 & 0.04 \\
& & & NBD & 1757.30 & 73.60 & 0.74 & 0.25 & 0.03  \\
\midrule
SyriaSizes & 4,075 & $\delta = 0.02$ & DP & 4175.70 & 66.04 & 0.65 & 0.17 & 0.01 \\
& & & PYP & 4234.30 & 68.55 & 0.64 & 0.19 & 0.01 \\
& & & NBNB & 4108.70 & 70.56 & 0.65 & 0.19 & 0.01 \\
& & & NBD & 3979.50 & 70.85 & 0.68 & 0.20 & 0.03  \\
\cmidrule{3-9}
& & $\delta = 0.05$ & DP & 4260.00 & 77.18 & 0.71 & 0.21 & 0.02 \\
& & & PYP & 4139.10 & 104.22 & 0.75 & 0.18 & 0.04 \\
& & & NBNB & 4047.10 & 55.18 & 0.73 & 0.20 & 0.04  \\
& & & NBD & 3863.90 & 68.05 & 0.75 & 0.22 & 0.07 \\
\cmidrule{3-9}
& & $\delta = 0.1$ & DP & 4507.40 & 82.27 & 0.80 & 0.19 & 0.03  \\
& & & PYP & 4540.30 & 100.53 & 0.80 & 0.20 & 0.03\\
& & & NBNB & 4400.60 & 111.91 & 0.80 & 0.23 & 0.03 \\
& & & NBD & 4251.90 & 203.23 & 0.82 & 0.25 & 0.04\\
\bottomrule
\end{tabular}
\label{tab:1}
\end{table}

\section{Summary}
\label{sec:disc}

Infinitely exchangeable clustering models assume that cluster sizes
grow linearly with the size of the data set. Although this assumption
is reasonable for some applications, it is inappropriate for
others. For example, when entity resolution is treated as a clustering
problem, the number of data points in each cluster should remain
small, even for large data sets. Applications like this require models
that yield clusters whose sizes grow sublinearly with the size of the
data set. We introduced the microclustering property as one way to
characterize models that address this requirement. We then introduced
a highly flexible class of models---KP models---that can exhibit this
property. We presented two models within this class---the NBNB model
and the NBD model---and showed that they are better suited to entity
resolution applications than two infinitely exchangeable clustering
models. We therefore recommend KP models for applications where the
size of each cluster should be unrelated to the size of the data set,
and each cluster should contain a negligible fraction of the total
number of data points.\looseness=-1

\subsubsection*{Acknowledgments}

We thank Tamara Broderick, David Dunson, Merlise Clyde, and Abel
Rodriguez for conversations that helped form the ideas in this
paper. In particular, Tamara Broderick played a key role in developing
the idea of microclustering. We also thank the Human Rights Data
Analysis Group for providing us with data. This work was supported in
part by NSF grants SBE-0965436, DMS-1045153, and IIS-1320219; NIH
grant 5R01ES017436-05; the John Templeton Foundation; the
Foerster-Bernstein Postdoctoral Fellowship; the UMass Amherst CIIR;
and an EPSRC Doctoral Prize Fellowship.\looseness=-1

\bibliographystyle{unsrt}
\bibliography{references}

\appendix

\section{Derivation of $P(C_N)$}
\label{sec:appendix_a}

In this appendix, we derive $P(C_N)$ for a general KP model, as well
as the NBNB and NBD models.

\subsection{KP Models}

We start with equation 2 and note that
\begin{equation}
  \label{eqn:CN}
  P(C_N) = P(C_N \g K)\,P(K),
\end{equation}
where $K = |C_N|$. To evaluate $P(C_N \g K)$, we need to sum over
all possible cluster assignments:
\begin{equation}
  P(C_N \g K) = \sum_{z_1, \ldots, z_N \in [K]} \underbrace{P(C_N \g
  z_1, \ldots, z_N, K)}_{I(z_1, \ldots, z_N \Rightarrow C_N)}\, P(z_1, \ldots, z_N \g K).
\end{equation}
Since $N_1, \ldots, N_K$ are completely determined by $K$
and $z_1, \ldots, z_N$, it follows that
\begin{align}
P(z_1, \ldots, z_N \g K) &= P(z_1, \ldots, z_N \g N_1, \ldots, N_K, K)\, P(N_1, \ldots, N_K \g K)\\
&= \frac{\prod_{k=1}^K N_k!}{N!} \prod_{k=1}^K P(N_k \g K)
\\
&= \frac{1}{N!} \prod_{k=1}^K N_k!\,\size{N_k}.
\end{align}
Therefore,
\begin{align}
P(C_N \g K) &= \sum_{z_1, \ldots, z_N \in [K]} I(z_1,\ldots, z_N \Rightarrow C_N) \,
     \frac{1}{N!} \prod_{k=1}^K N_k!\,\size{N_k}\\
&= \frac{1}{N!} \left( \prod_{c\in C_N} |c|!\, \size{|c|} \right)
    \sum_{z_1, \ldots, z_N \in [K]} I(z_1, \ldots, z_N \Rightarrow C_N)\\
&= \frac{K!}{N!} \prod_{c\in C_N} |c|!\, \size{|c|}. \label{eqn:CN_g_K}
\end{align}
Substituting equation~\ref{eqn:CN_g_K} into equation~\ref{eqn:CN} and using $K\sim\clusters$ we obtain
\begin{equation}
\label{eq:FMMC_final_exp}
  P(C_N) = \frac{|C_N|!\,\cluster{|C_N|}}{N!} \left(\prod_{c\in C_N} |c|!\, \size{|c|}\right).
\end{equation}

\subsection{The NBNB Model}

For fixed values of $r$ and $p$, the NBNB model is a specific case of
a KP model with
\begin{equation}
  \label{eq:NBNB_both}
\cluster{k}=\frac{\Gamma\left(k+a\right) q^k\, (1-q)^{a}}{(1-(1-q)^a)\,\Gamma\left(a\right)k!}
\quad \textrm{and} \quad
\size{m}=\frac{\Gamma\left(m+r\right)p^m\,
(1-p)^{r}}{(1-(1-p)^r)\,\Gamma\left(r\right)m!}\,,
\end{equation}
for $k$ and $m$ in $\mathcal{N} = \{1,2,\dots\}$. Combining
equations~\ref{eq:FMMC_final_exp} and~\ref{eq:NBNB_both} gives
\begin{align}
P(C_N \g a, q, r, p)
&=
\frac{|C_N|!}{N!}\,
\frac{\Gamma\left(|C_N|+a\right)q^{|C_N|}\,
(1-q)^{a}}{(1-(1-q)^a)\,\Gamma\left(a\right)|C_N|!}
\prod_{c\in C_N} |c|!\,\frac{\Gamma\left(|c|+r\right)p^{|c|}\,
(1-p)^{r}}{(1-(1-p)^r)\,\Gamma\left(r\right)|c|!}\\
&= \frac{\Gamma\left(|C_N|+a\right)q^{|C_N|}\,
(1-q)^{a}}{N!\,(1-(1-q)^a)\,\Gamma\left(a\right)}
\prod_{c\in C_N} \frac{\Gamma\left(|c|+r\right)p^{|c|}\,
(1-p)^{r}}{(1-(1-p)^r)\,\Gamma\left(r\right)}.\label{eq:NBNB_appendix}
\end{align}

Conditioning on $N$ and removing constant terms, we obtain
\begin{align}\label{eq:NBNB_appendix_N}
P(C_N \g N,a,q, r,p)
&\propto\Gamma\left(|C_N|+a\right)\beta^{|C_N|}
\prod_{c\in C_N} \frac{\Gamma\left(|c|+r\right)}{\Gamma\left(r\right)},
\end{align}
where $\beta=\frac{q\,(1-p)^{r}}{1-(1-p)^r}$.
Equation \ref{eq:NBNB_appendix_N} leads to the following
reseating algorithm:
\begin{itemize}
\item for $n=1,\ldots, N$, reassign element $n$ to
\begin{itemize}
\item an existing cluster $c\in C_{N} \!\setminus\! n$ with probability $\propto
|c|+r$
\item or a new cluster with probability $\propto \left(|C_N \!\setminus\! n|+a\right)\beta$.
\end{itemize}
\end{itemize}

Adding the prior terms for $r$ and $p$ to equation
\ref{eq:NBNB_appendix} we obtain the joint distribution of $C_N$, $r$
and $p$:
\begin{align}
&P(C_N,r,p \g a, q, \eta_r, s_r,\betaA_p,\betaB_p )\notag\\
&\quad =P(r \g \eta_r, s_r)\,P(p \g \betaA_p,\betaB_p)\,P(C_N \g r,p)\\
&\quad=
\frac{r^{\eta_r-1}e^{-\frac{r}{s_r}}}{\Gamma\left(\eta_r\right)s_r^{\eta_r}}\,
\frac{p^{\betaA_p-1}(1-p)^{\betaB_p-1}}{B(\betaA_p,\betaB_p)}\times
{} \notag\\
&\quad\quad\frac{\Gamma\left(|C_N|+a\right)q^{|C_N|}\,
(1-q)^{a}}{N!\,(1-(1-q)^a)\,\Gamma\left(a\right)}
\prod_{c\in C_N} \frac{\Gamma\left(|c|+r\right)p^{|c|}\,
(1-p)^{r}}{(1-(1-p)^r)\,\Gamma\left(r\right)}\\
&\quad\propto
r^{\eta_r-1}\,e^{-\frac{r}{s_r}}\,p^{N+\betaA_p-1}\,(1-p)^{\betaB_p-1}
\left(\frac{q\,(1-p)^{r}}{1-(1-p)^r}\right)^{|C_N|} \times{}\notag\\
&\quad\quad\frac{\Gamma\left(|C_N|+a\right)}{N!}
\prod_{c\in C_N} \frac{\Gamma\left(|c|+r\right)}{\Gamma\left(r\right)}.
\end{align}

Therefore, the conditional posterior distributions over $r$ and $p$ are
\begin{align}
P(r \g C_N,p,\eta_r, s_r)  &\propto
\frac{r^{\eta_r-1}\,e^{-\frac{r}{s_r}}\,(1-p)^{r\,|C_N|}}{(1-(1-p)^r)^{|C_N|}}
\prod_{c\in C_N}
\frac{\Gamma\left(|c|-1+r\right)}{\Gamma\left(r\right)}\\
P(p \g C_N, r, \betaA_p, \betaB_p)  &\propto
\frac{p^{N+\betaA_p-1}(1-p)^{r\,|C_N|+\betaB_p-1}}{(1-(1-p)^r)^{|C_N|}}.
\end{align}

\subsection{The NBD Model}

For fixed $\sizes$, the NBD model is a specific case of a KP
model. Therefore, 
\begin{align}
P(C_N \g a, q, \sizes)
&=\frac{\Gamma\left(|C_N|+a\right)q^{|C_N|}\,
(1-q)^{a}}{N!\,(1-(1-q)^a)\Gamma\left(a\right)}
\prod_{c \in C_N} |c|!\,\size{|c|}.
\end{align}
Conditioning on $N$ and removing constant terms, we obtain
\begin{align*}
P(C_N \g N, a, q, \sizes)
&\propto\Gamma\left(|C_N|+a\right)q^{|C_N|}
\prod_{c\in |C_N|} |c|!\,\size{|c|}.
\end{align*}
Via Dirichlet--multinomial conjugacy,
\begin{equation}
  \label{eqn:sizes_cond_posterior2}
\sizes \g C_N, \alpha, \boldsymbol{\mu}^{(0)} \sim
\textrm{Dir}\left(\alpha\,\size{1}^{(0)}+L_{1},\alpha\,\size{2}^{(0)}+L_{2},\dots\right),
\end{equation}
where $L_m$ is the number of clusters of size $m$ in $C_N$. Although
$\sizes$ is an infinite-dimensional vector, only the first $N$
elements affect $P(C_N \g a, q, \sizes)$. Therefore, it is sufficient
to sample the $(N+1)$-dimensional vector $(\size{1}, \ldots, \size{N},
1 - \sum_{m=1}^N \size{m})$ from
equation~\ref{eqn:sizes_cond_posterior2}, modified accordingly:
\begin{align}
&(\size{1}, \ldots, \size{N}, 1 - \sum_{m=1}^N \size{m}) \g C_N,
  \alpha, \size{1}^{(0)}, \ldots, \size{N}^{(0)} \notag\\
&\quad \sim \textrm{Dir}\left(\alpha\,\size{1}^{(0)}+L_{1}, \ldots,
\alpha\,\size{N}^{(0)}+L_{N},\alpha\left(1 - \sum_{m=1}^N \mu_m^{(0)}\right)\right).
\end{align}
We can then discard $1 - \sum_{m=1}^N \size{m}$.

\section{Proof of the Microclustering Property for a Variant of the NBNB Model}
\label{sec:appendix_micro_proof}

\begin{theorem}
\label{theorem:NBNB-micro} If $C_N$ is drawn from a KP model with $\clusters = \textrm{NegBin}\left(a,
q\right)$ and $\sizes = \textrm{NegBin}\left(r, p\right)$,\footnote{We
  have not truncated the negative binomial distributions, so this is a
  minor variant the NBNB model.}  then for all $\epsilon > 0$, $P(M_N
\,/\, N \geq \epsilon) \rightarrow 0$ as $N \rightarrow \infty$, where
$M_N$ is the size of the largest cluster in $C_N$.
\end{theorem}

In this appendix, we provide a proof of
theorem~\ref{theorem:NBNB-micro}.

We use the following fact: $\Gamma\left(x+a\right) /\,
\Gamma\left(x\right) \asymp x^a$ as $x\rightarrow\infty$ for any
$a\in\R$ via Stirling's approximation to the gamma function. We use
$f(x) \asymp g(x)$ to denote that $f(x)\,/\,g(x)\rightarrow 1$ as
$x\rightarrow\infty$.\looseness=-1

\begin{lemma}
\label{lemma:pks-diverges}
For any $k\in\{1,2,\ldots\}$, $P(K=k \g N=n)\rightarrow 0$ as $n\rightarrow\infty$.
\end{lemma}

\begin{proof} Because $N \g K = k \sim \textrm{NegBin}\left(k r,
  p\right)$,
  \begin{equation*}
P(K=k, N=n) =
  \frac{\Gamma\left(k+a\right)}{k!\,\Gamma\left(a\right)}\, (1-q)^a\,
  q^k\, \frac{\Gamma\left(n + k r\right)}{n!\,\Gamma\left(k
    r\right)}\, (1-p)^{k r}\, p^n.
  \end{equation*}
Via the fact noted above, $\Gamma\left(n + k r\right) /\,\Gamma\left(n
+ k r + r\right) \asymp 1\,/\,(n + k r)^r \rightarrow 0$ as
$n\rightarrow\infty$, so
\begin{equation*}
\frac{P(K=k,N=n)}{P(K=k+1,N=n)} 
=
\frac{\Gamma\left(k+a\right)(k+1)}{\Gamma\left(k+a+1\right)q}\,
\frac{\Gamma\left(kr+r\right)}{\Gamma\left(kr\right)}\,
\frac{\Gamma\left(n+kr\right)}{\Gamma\left(n+kr+r\right)}
\rightarrow 0 \textrm{ as }n\rightarrow\infty.
\end{equation*}
Therefore,
\begin{equation*}
P(K=k\g N=n) = \frac{P(K=k,N=n)}{\sum_{k'=0}^\infty P(K=k',N=n)} \leq
\frac{P(K=k,N=n)}{P(K=k+1,N=n)} \rightarrow 0.
\end{equation*}
\end{proof}

\begin{lemma}
\label{lemma:bounding-sequence}
For any $\epsilon \in (0,1)$, there exist $c_1,c_2, \ldots \geq 0$,
not depending on $n$, such that $c_k \rightarrow 0$ as $k \rightarrow
\infty$ and $ k\, P(N_1\,/\,n \geq \epsilon \g K=k,N=n) \leq c_k $ for
all $n \geq 2\,/\,\epsilon$ and $k\in\{1,2,\ldots\}$.
\end{lemma}

Before proving lemma~\ref{lemma:bounding-sequence}, we first show how
theorem~\ref{theorem:NBNB-micro} follows from it.

\begin{proof}[Proof of theorem~\ref{theorem:NBNB-micro}]
Let $\epsilon \in (0,1)$ and choose $c_1,c_2,\ldots$ by
lemma~\ref{lemma:bounding-sequence}. For any $n \geq 2\,/\,\epsilon$,
\begin{align*}
&P(M_n \,/\, n \geq \epsilon \g N=n) \\
&\quad=\sum_{k=1}^\infty P(N_1\,/\,n\geq\epsilon \text{ or } \cdots \text{ or } N_K\,/\,n\geq\epsilon \g K=k,N=n)\,P(K=k\g N=n) \\
&\quad\leq\sum_{k=1}^\infty\sum_{i=1}^k P(N_i\,/\,n\geq\epsilon \g K=k,N=n)\,P(K=k\g N=n) \\
&\quad=\sum_{k=1}^\infty k\, P(N_1\,/\,n\geq\epsilon \g
  K=k,N=n)\,P(K=k\g N=n) \\
&\quad\leq\sum_{k=1}^\infty c_k\, P(K=k \g N=n)
  \leq \textrm{sup}\,\{c_k : k>m\} + \sum_{k=1}^m c_k\, P(K=k\g N=n)
\end{align*}
for any $m\geq 1$. (We note that we only summed over $k\geq 1$ because
$P(K=0 \g N=n) = 0$ for any $n\geq 1$.)  Therefore, via lemma
\ref{lemma:pks-diverges}, $\textrm{limsup}_n P(M_n \,/\, n \geq
\epsilon \g N=n) \leq \textrm{sup}\,\{c_k : k>m\}$. Finally, because
$\textrm{sup}\,\{c_k : k>m\}\rightarrow 0$ as $m\rightarrow\infty$,
theorem~\ref{theorem:NBNB-micro} follows directly from
lemma~\ref{lemma:bounding-sequence}, as desired.\looseness=-1
\end{proof}

To prove lemma~\ref{lemma:bounding-sequence}, we need two supporting
results.

\begin{lemma}
\label{lemma:beta-limit}
If $b>(r+1)\,/\,r$ and
$\theta_k\sim\textrm{Beta}\left(r,(k-1)\,r\right)$, then $k\,
P(\theta_k \geq \frac{b \log{(k)}}{k}) \rightarrow 0$ as
$k\rightarrow\infty$.
\end{lemma}
\begin{proof}
Let $a_k = (b\log{(k)})\,/\,k$, and suppose that $k$ is large enough
that $a_k\in(0,1)$. First, for any $\theta\in(a_k,1)$, we have
$\theta^{r-1}\leq 1\,/\,a_k$.  Second, $B\left(r,(k-1)\,r\right) =
\Gamma\left(r\right)\Gamma\left(k r - r\right)
/\,\Gamma\left(kr\right) \asymp \Gamma\left(r\right)(k r)^{-r}$ as
$k\rightarrow\infty$, via Stirling's approximation, as we noted
previously. Third, because $1+x\leq \exp{(x)}$ for any $x\in\R$,
$(1-a_k)^{k r} \leq \exp{(-a_k)}^{k r} = k^{- r b}$.  Therefore, we
obtain
\begin{align*}
&k\, P(\theta_k\geq a_k)\\
  &\quad= \frac{k}{B\left(r,(k-1)\,r\right)}
\int_{a_k}^1 \theta^{r-1}\,(1-\theta)^{(k-1)\,r-1}\ \textrm{d}\theta
\\
&\quad\leq \frac{k\,/\,a_k}{B\left(r,(k-1)\,r\right)} \int_{a_k}^1
(1-\theta)^{(k-1)r-1}\ \textrm{d}\theta = \frac{k\,/\,a_k}{B\left(r,(k-1)\,r\right)}\,\frac{(1-a_k)^{(k-1)\,r}}{(k-1)\,r}
\\
&\quad\leq \frac{k\,/\,a_k}{B\left(r,(k-1)\,r\right)}\, \frac{k^{-r b}\,(1-a_k)^{-r}}{(k-1)\,r}
\asymp \frac{k^2 \,/\, (b \log{(k)})}{\Gamma\left(r\right)(k r)^{-r}}\,
\frac{k^{-r b}}{k r} =
\frac{r^{r-1}\,k^{-br+r+1}}{
\Gamma\left(r\right)(b\log{(k)})
}
 \rightarrow 0
\end{align*}
as $k\rightarrow 0$ because $b > (r+1)\,/\,r$.
\end{proof}

\begin{lemma}
\label{lemma:binomial-bound}
Let $b>0$ and $\epsilon\in(0,1)$, as well as $k>1$ and
$n\in\{1,2,\ldots\}$. If $(b\log{(k)})\,/\,k < 1$ and
$X\sim\textrm{Bin}\left(n,(b\log{(k)}\right) /\,k)$, then $P(X \geq
n\epsilon) \leq (1 + b\log{(k)})^n \,/\, k^{n\epsilon}$.
\end{lemma}

\begin{proof}
Let $Z_1,\ldots,Z_n \iid \textrm{Bern}\left((b\log{(k)})\,/\,k\right)$.
Because $x \mapsto k^x$ is strictly increasing,
\begin{align*}
    P(X\geq n\epsilon) = P(k^X \geq k^{n\epsilon}) \leq \frac{\mathbb{E}[k^X]}{k^{n\epsilon}}
= \frac{\prod_{i=1}^n \mathbb{E}[k^{Z_i}]}{k^{n\epsilon}} \leq \frac{(1+b\log{(k)})^n}{k^{n\epsilon}}
\end{align*}
via Markov's inequality.
\end{proof}

\begin{proof}[Proof of lemma~\ref{lemma:bounding-sequence}]
First, let $\epsilon\in(0,1)$.  Next, let $b=(r+2)\,/\,r$ and choose
$k^*\in\{2,3,\ldots\}$ to be sufficiently large that
$(1+b\log{(k)})\,/\,k^{\epsilon} < 1$ and $(b\log{(k)})\,/\,k <
\epsilon$ for all $k \geq k^*$.  Then, for $k=1,2,\ldots,{k^*-1}$,
define $c_k=k$, and, finally, for $k=k^*,k^*+1,\ldots$, define
\begin{equation*}
c_k = k^{-1} (1 + b\log{(k)})^{2 / \epsilon} + k\, P\left(\theta_k \geq
\frac{b\log{(k)}}{k}\right),
\end{equation*}
where $\theta_k\sim\textrm{Beta}\left(r,(k-1)\,r\right)$.

Via lemma~\ref{lemma:beta-limit}, $c_k \rightarrow 0$ as
$k\rightarrow\infty$. Trivially, for $k < k^*$, $k\,
P(N_1\,/\,n\geq\epsilon\g K=k,N=n) \leq k = c_k$.

Let $k \geq k^*$ and suppose that $n\geq 2\,/\,\epsilon$.  Via a
straightforward calculation, we can show that $N_1 \g K=k,N=n\sim
\textrm{BetaBin}\left(n,r,(k-1)\,r\right)$. (This follows from the
fact that if $Y\sim\textrm{NegBin}\left(r,p\right)$ and,
independently, $Z\sim\textrm{NegBin}\left(r',p\right)$, then $Y\g
(Y+Z)=n \sim \textrm{BetaBin}\left(n,r,r'\right)$.)  Therefore, if we
define $\theta \sim \textrm{Beta}\left(r,(k-1)\,r\right)$, $X \g
\theta\sim\textrm{Bin}\left(n,\theta\right)$, and $a =
(b\log{(k)})\,/\,k$, then we have
\begin{align*}
k\, P(N_1\,/\,n\geq\epsilon\g K=k,N=n) &= k\, P(X\geq n\epsilon)\\
&= k\, P(X\geq n\epsilon,\theta < a) + k\, P(X\geq n\epsilon,\theta\geq a).
\end{align*}
However, $k\, P(X\geq n\epsilon,\theta\geq a) \leq k\, P(\theta\geq a)
= k\, P\left(\theta_k\geq \frac{b\log{(k)}}{k}\right)$. To handle the first term,
we note that as a function of $\theta$, $P(X=x\g \theta)$ is
nondecreasing on $(0,\epsilon)$ whenever $x\,/\,n\geq\epsilon$ because
$\frac{\textrm{d}P(X=x\g
  \theta)}{\textrm{d}\theta}={{n}\choose{x}}\,\theta^{x-1}\,(1-\theta)^{n-x-1}\,(x-n\theta)$.
Therefore, $P(X\geq n\epsilon \g \theta) = \sum_{x \geq n\epsilon}
P(X=x\g\theta)$ is nondecreasing on $(0,\epsilon)$. Finally, because
our choice of $k^*$ means that $a\in(0,\epsilon)$,
\begin{align*}
 &k\, P(X\geq n\epsilon,\theta < a)\notag\\
&\quad= k \int_0^a P(X\geq n\epsilon \g \theta)\, P(\theta)\ \textrm{d}\theta
\leq k\, P(X\geq n\epsilon \g \theta = a) \\ 
&\quad\leq k\,(1 + b\log{(k)})^n \,/\, k^{n\epsilon} = k\left(\frac{1+b\log{(k)}}{k^\epsilon}\right)^n\\
&\quad\leq k\left(\frac{1+b\log{(k)}}{k^\epsilon}\right)^{2/\epsilon} = k^{-1}\left(1+b\log{(k)}\right)^{2/\epsilon},
\end{align*}
where the second inequality is via lemma~\ref{lemma:binomial-bound}
and the third inequality holds because $n\geq 2\,/\,\epsilon$ and
$(1+b\log{(k)})\,/\,k^\epsilon < 1$ because of our choice of $k^*$.
Thus, $k\, P(N_1\,/\,n\geq\epsilon\g K=k,N=n) \leq c_k$.\looseness=-1
\end{proof}

This completes the proof of theorem~\ref{theorem:NBNB-micro}.

\section{The Chaperones Algorithm}
\label{sec:appendix_b}

For large data sets with many small clusters, standard Gibbs sampling
algorithms (such as the one outlined in section 3) are too slow. In
this appendix, we therefore propose a new Gibbs-type sampling
algorithm, which we call the chaperones algorithm. This algorithm is
inspired by existing split--merge Markov chain sampling
algorithms~\cite{jain04split--merge,steorts??bayesian,steorts14smered};
however, it is simpler, more efficient, and---most
importantly---likely exhibits better mixing properties when there are
many small clusters.

In a standard Gibbs sampling algorithm, each iteration involves
reassigning each data point $x_n$ for $n = 1, \ldots, N$ to an
existing cluster or to a new cluster by drawing a sample from $P(C_N
\g N, C_N \!\setminus\! n, x_1, \ldots, x_N)$. When the number of clusters
is large, this step can be inefficient because the probability that
element $n$ will be reassigned to a given cluster will, for most
clusters, be extremely small.\looseness=-1

The chaperones algorithm focuses on reassignments that have higher
probabilities. If $c_n \in C_N$ denotes the cluster containing data
point $x_n$, then each iteration consists of the following steps:
\begin{enumerate}
\item Randomly choose two chaperones, $i, j \in \{1, \ldots, N\}$ from
  a distribution $P(i, j \g x_1, \ldots, x_N)$ where the probability
  of $i$ and $j$ given $x_1, \ldots, x_N$ is greater than zero for all
  $i \neq j$. This distribution must be independent of the current
  state of the Markov chain $C_N$; however, crucially, it may depend
  on the observed data points $x_1, \ldots, x_N$.

\item Reassign each $x_n \in c_i \cup c_j$ by sampling from $P(C_N \g N,
  C_N \!\setminus\! n, c_i \cup c_j, x_1, \ldots, x_N)$.
\end{enumerate}

In step 2, we condition on the current partition of all data points
except $x_n$, as in a standard Gibbs sampling algorithm, but we also
force the set of data points in $c_i \cup c_j$ to remain
unchanged---i.e., $x_n$ must remain in the same cluster as at least
one of the chaperones. (If $n$ is a chaperone, then this requirement
is always satisfied.) In other words, we view the non-chaperone data
points in $c_i \cup c_j$ as ``children'' who must remain with a
chaperone at all times. Step 2 is almost identical to the restricted
Gibbs moves found in existing split--merge algorithms, except that the
chaperones $i$ and $j$ can also change clusters, provided they do not
abandon any of their children. Splits and merges can therefore occur
during step 2: splits occur when one chaperone leaves to form its own
cluster; merges occur when one chaperone, belonging to a singleton
cluster, then joins the other chaperone's cluster.

The chaperones algorithm can be justified as follows: For any fixed
pair of chaperones $(i,j)$, step 2 is a sequence of Gibbs-type moves
and therefore has the correct stationary distribution. Randomly
choosing the chaperones in step 1 amounts to a random move, so, taken
together, steps 1 and 2 also have the correct stationary distribution
(see, e.g., \cite{tierney94markov}, sections 2.2 and 2.4). To
guarantee irreducibility, we start by assuming that $P(x_1, \ldots,
x_N \g C_N)\,P(C_N) > 0$ for any $C_N$ and by letting $C'_N$ denote
the partition of $N$ in which every element belongs to a singleton
cluster. Then, starting from any partition $C_N$, it is easy to check
that there is a positive probability of reaching $C'_N$ (and vice
versa) in finitely many iterations; this depends on the assumption
that $P(i, j \g x_1, \ldots, x_N) > 0$ for all $i \neq
j$. Aperiodicity is also easily verified since the probability of
staying in the same state is positive.

The main advantage of the chaperones algorithm is that it can exhibit
better mixing properties than existing sampling algorithms. If the
distribution $P(i, j \g x_1, \ldots, x_N)$ is designed so that $x_i$
and $x_j$ tend to be similar, then the algorithm will tend to consider
reassignments that have a relatively high probability. In addition,
the algorithm is easier to implement and more efficient than existing
split--merge algorithms because it uses Gibbs-type moves, rather than
Metropolis-within-Gibbs moves.\looseness=-1

\newpage
\section{The Syria2000 and SyriaSizes Data Sets}
\label{sec:syrianData}

\begin{figure}[htbp]
\begin{center}
\includegraphics[scale=0.3]{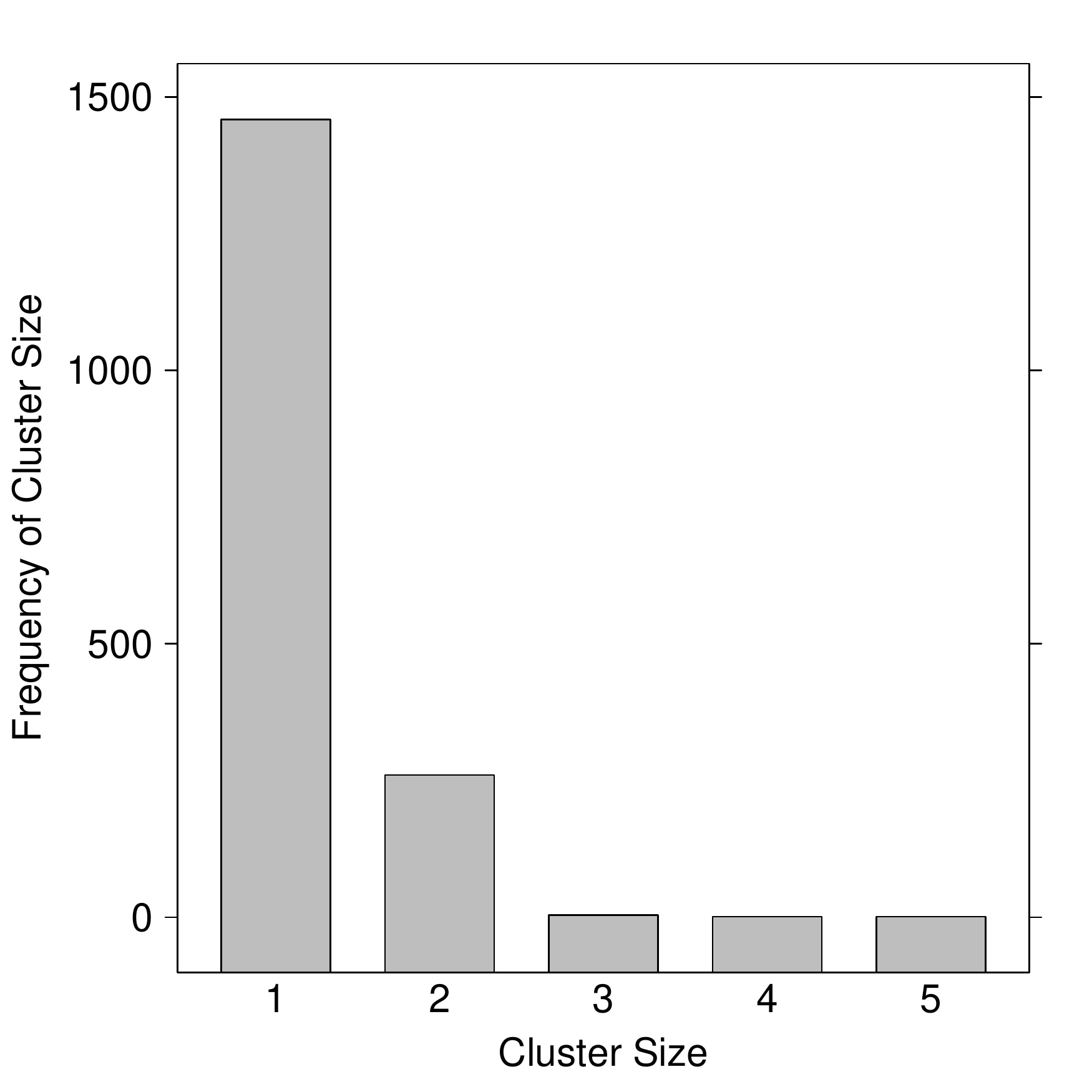}
\includegraphics[scale=0.3]{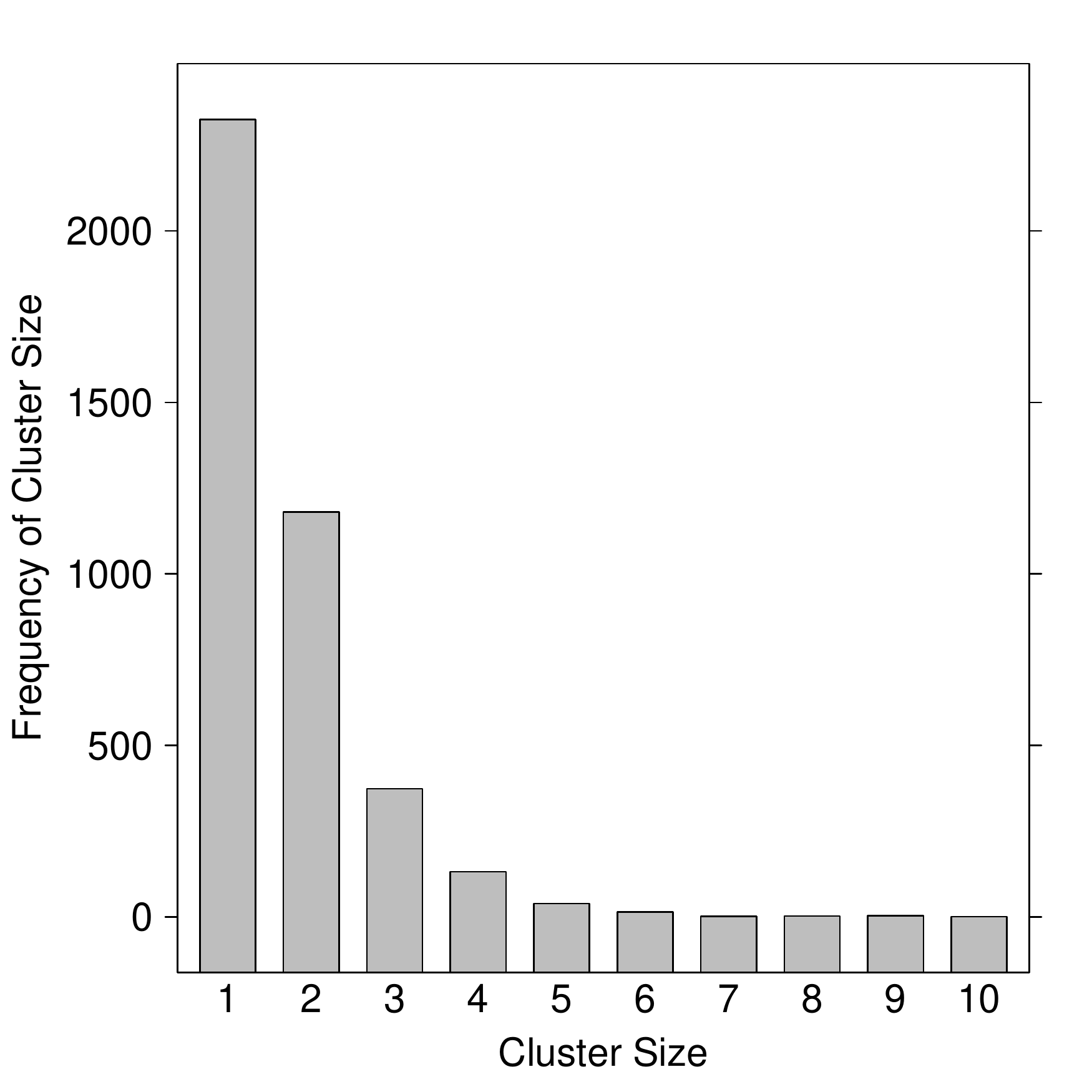}
\caption{Cluster size distributions for the Syria2000 (left)
  and SyriaSizes (right) data sets.}
\label{fig:SyriaBarCharts}
\end{center}
\end{figure}

\end{document}